%% file: main.tex
\documentclass[11pt]{llncs}

\usepackage[english]{babel}
\usepackage[T1]{fontenc}
\usepackage[latin1]{inputenc}
\usepackage{times}
\usepackage{amsfonts,amssymb,amsmath,amscd,latexsym,graphicx}

\newcommand{\set}[1]{\left\{ #1 \right\}}

\newcommand{\bin}{\mathrm{bin}}

\newcommand{\G}{\mathcal{G}}

\newcommand{\Alt}[1]{\mathrm{Alt}\left( #1 \right)}
\newcommand{\ND}[1]{\mathrm{NonDet}\left( #1 \right)}
\newcommand{\Det}[1]{\mathrm{Det}\left( #1 \right)}
\newcommand{\Reg}{\mathrm{Reg}}

\renewcommand{\L}{\mathcal{L}}

\newcommand{\Eq}{\textsc{CountEq}}
\newcommand{\NotEq}{\textsc{NotEq}}
\newcommand{\Leq}{\textsc{Lexicographic}}
\newcommand{\Prime}{\textsc{Primes}}

\newcommand{\M}{\mathcal{M}}

\newcommand{\N}{\mathbb{N}}
\newcommand{\Z}{\mathbb{Z}}
\newcommand{\A}{\mathcal{A}}
\newcommand{\B}{\mathcal{B}}

\renewcommand{\P}{\mathcal{P}}

\begin{document}

\title{Lower Bounds for Alternating Online State Complexity\thanks{This work was done in part while the author was visiting the Simons Institute for the Theory of Computing.}}

\author{Nathana\"el Fijalkow}

\institute{University of Oxford, United Kingdom}

\maketitle

\begin{abstract}
The notion of Online State Complexity, introduced by Karp in 1967, quantifies the amount of states required to solve a given problem using an online algorithm,
which is represented by a \textit{deterministic} machine scanning the input from left to right in one pass.

In this paper, we extend the setting to \textit{alternating} machines as introduced by Chandra, Kozen and Stockmeyer in 1976:
such machines run independent passes scanning the input from left to right and gather their answers through boolean combinations.

We devise a lower bound technique relying on boundedly generated lattices of languages, and give two applications of this technique.
The first is a hierarchy theorem, stating that the polynomial hierarchy of alternating online state complexity is infinite, 
and the second is a linear lower bound on the alternating online state complexity of the prime numbers written in binary.
This second result strengthens a result of Hartmanis and Shank from 1968, which implies an exponentially worse lower bound for the same model.
\end{abstract}



\section{Online State Complexity}
\label{sec:online_stace_complexity}
\input{intro}

\subsection{Definitions}
\input{definitions}

\subsection{Applications to One-way Realtime Turing Machines}
\input{turing_machines}

\subsection{Related Works}
\input{related_works}

\subsection{Complexity Classes and Examples}
\input{complexity_classes}

\section{A Lower Bound Technique}
\label{sec:lower_bound}

In this section, we develop a generic lower bound technique for alternating online state complexity.
It is based on the size of generating families for some lattices of languages;
we describe it in Subsection~\ref{subsec:lattices}, and a concrete approach to use it,
based on query tables, is developed in Subsection~\ref{subsec:query}.
We apply it on an example in Subsection~\ref{subsec:lower_bound_example}.

\subsection{Boundedly Generated Lattices of Languages}
\label{subsec:lattices}
\input{lower_bound}

\subsection{The Query Table Method}
\label{subsec:query}
\input{query_tables}

\subsection{A First Application of the Query Table Method}
\label{subsec:lower_bound_example}
\input{non-det-cfl}

\section{A Hierarchy Theorem for Languages of Polynomial Alternating Online State Complexity}
\label{sec:hierarchy}

\input{alt}

\section{The Online State Complexity of Prime Numbers}
\label{sec:prime}

\input{prime}

\section*{Conclusion}

We have developed a generic lower bound technique for alternating online state complexity, and applied it to two problems.
The first result is to show that the polynomial hierarchy of alternating online algorithms is infinite.
The second result is to give lower bounds on the alternating online state complexity of the language of prime numbers;
we show that it is not sublinear, which is an exponential improvement over the previous result.
However, the exact complexity is left open; we conjecture that it is not subexponential, but obtaining this result may require
major advances in number theory.


\bibliographystyle{alpha}
\bibliography{bib}

\end{document}

%% file: intro.tex
An \textit{online algorithm} has a restricted access to its input: it scans it exactly once from left to right.
The notion of \textit{online computing} has been identified as a fundamental research question in the 80s, 
and has since then blossomed into several directions with various approaches.
In this work we are concerned with \textit{complexity} questions, and in particular about the use of \textit{space} for online algorithms,
harkening back to a series of works initiated by Karp~\cite{Karp67}.
An impressive result in this line of work is a tight bound on the complexity of checking the primality of a number 
written in binary for \textit{deterministic} online algorithms by Hartmanis and Shank~\cite{HartmanisShank69}.

\vskip1em
The class of \textit{deterministic} online algorithms being rather weak, two natural extensions are commonly studied:
adding randomisation or allowing the algorithm to make several passes over the input.
The latter idea is popular in the field of \textit{streaming algorithms}, introduced by Munro and Paterson~\cite{MP80}, Flajolet and Martin~\cite{FM85}, 
and Alon, Matias and Szegedy~\cite{AMS96}.

We initiate in this paper the study of \textit{alternating} machines for representing online algorithms, 
as an expressive class of online algorithms making several passes over the input.
Alternating (Turing) machines have been introduced by Chandra, Kozen and Stockmeyer~\cite{ChandraStockmeyer76,Kozen76,ChandraKozenStockmeyer81},
and generalise non-deterministic models where along the computation, the machine makes guesses about the input, 
and the computation is successful if \textit{there exists} a sequence of correct guesses.
In other words, these guesses are \textit{disjunctive choices}; 
the alternating model restores the symmetry by introducing disjunctive and conjunctive choices,
resolved by two competing agents, a Prover and a Verifier.
Intuitively, the Prover argues that the input is valid and the Verifier challenges this claim.

Alternating machines are very expressive: the seminal results of Chandra, Kozen and Stockmeyer
state that alternating Turing machines are exponentially more expressive than deterministic ones,
which materialises by complexity classes equalities such as $\text{APTIME} = \text{PSPACE}$ and $\text{ALOGSPACE} = \text{PTIME}$.

\vskip1em
\textbf{Contributions and organisation of the paper.}
We study alternating online state complexity, \textit{i.e.} the amount of states required to solve a given problem using an alternating machine.
We devise a generic lower bound technique based on boundedly generated lattices of languages.

We give the basic definitions and show some examples in the remainder of this section.
We describe our lower bound technique in Section~\ref{sec:lower_bound}, and give two applications:
\begin{itemize}
	\item \textit{Hierarchy theorem}: in Section~\ref{sec:hierarchy}, we show a hierarchy theorem: 
	for each natural number $\ell$ greater than or equal to $2$, 
	there exists a language having alternating online state complexity $n^{\ell}$ but not $n^{\ell - \varepsilon}$ for any $\varepsilon > 0$.

	\item \textit{Prime numbers}: in Section~\ref{sec:prime}, we look at the language of prime numbers written in binary.
	The works of Hartmanis and Shank culminated in showing that it does not have subexponential \textit{deterministic} online state complexity~\cite{HartmanisShank69}.
	We consider the stronger model of \textit{alternating} online algorithms, 
	and first observe that Hartmanis and Shank's techniques imply a \textit{logarithmic} lower bound 
	on the \textit{alternating} online state complexity.
	Our contribution is to strengthen this result by showing a \textit{linear} lower bound, which is thus an exponential improvement.
\end{itemize}

%% file: definitions.tex
We fix an \textit{input alphabet} $A$, which is a finite set of letters.
A \textit{word} is a finite sequence of letters, often denoted $w = w(0) w(1) \cdots w(n-1)$, 
where $w(i)$'s are letters from the alphabet $A$, \textit{i.e.} $w(i) \in A$.
We say that $w$ has length $n$, and denote it $|w|$.
The empty word is denoted $\varepsilon$.
We denote $A^*$ the set of all words and $A^{\le n}$ the set of words of length at most $n$.

For a set $E$, we denote $\B^+(E)$ the set of positive boolean formulae over $E$.
For instance, if $E = \set{p,q,r}$, an element of $\B^+(E)$ is $p \wedge (q \vee r)$.

\vskip1em
Our aim is to prove lower bounds on the state complexity of online algorithms.
Following Karp~\cite{Karp67}, we do not work with Turing machines but with a more general model that we simply call machines;
since we are interested in lower bounds, this makes our results stronger.
We define alternating machines following Chandra, Kozen and Stockmeyer~\cite{ChandraStockmeyer76,Kozen76,ChandraKozenStockmeyer81}.

\begin{definition}[Alternating Machines]
An alternating machine is given by a (potentially infinite) set $Q$ of states, an initial state $q_0 \in Q$,
a transition function $\delta : Q \times A \to \B^+(Q)$ and a set of accepting states $F \subseteq Q$.
\end{definition}

To define the semantics of alternating machines, we use acceptance games.
Consider an alternating machine $\A$ and an input word $w$, 
we define the acceptance game $\G_{\A,w}$ as follows: it has two players, Prover and Verifier.
The Prover claims that the input word $w$ should be accepted, and the Verifier challenges this claim.

The game starts from the initial state $q_0$, and with each letter of $w$ read from left to right,
a state is chosen through the interaction of the two players.
If in a state $q$ and reading a letter $a$, the new state
is obtained using the boolean formula $\delta(q,a)$; 
Prover chooses which formula is satisfied in a disjunction, and Verifier does the same for conjunctions.
A play is won by Prover if it ends up in an accepting state.

The input word $w$ is accepted by $\A$ if Prover has a winning strategy in the acceptance game $\G_{\A,w}$.
The language recognised by $\A$ is the set of input words accepted by $\A$.

\vskip1em
As special cases, a machine is:
\begin{itemize}
	\item \textit{non-deterministic} if for all $q$ in $Q$, $a$ in $A$, $\delta(q,a)$ is a disjunctive formula,
	\item \textit{universal} if for all $q$ in $Q$, $a$ in $A$, $\delta(q,a)$ is a conjunctive formula,
	\item \textit{deterministic} if for all $q$ in $Q$, $a$ in $A$, $\delta(q,a)$ is an atomic formula,
	\textit{i.e.} if $\delta : Q \times A \to Q$.
\end{itemize}

\begin{definition}[Online State Complexity Classes]
Let $f : \N \to \N$.
The language $L$ is in $\Alt{f}$ if there exists an alternating machine recognising $L$ and a constant $C$ such that for all $n$ in $\N$:
\[
\left|\set{q \in Q \mid \exists w \in A^{\le n}, q \text{ appears in the game } \G_{\A,w}}\right| \le C \cdot f(n).
\]
Such a machine is said to use $f$ many states.

Similarly, we define $\ND{f}$ for non-deterministic machines and $\Det{f}$ for deterministic machines.
\end{definition}

For the sake of succinctness, the acronym OSC will be used in lieu of online state complexity.

We denote the function $f : n \mapsto f(n)$ by $f(n)$, so for instance $\Alt{\log(n)}$ denotes the languages having logarithmic OSC.

We say that $L$ has sublogarithmic (respectively sublinear) alternating OSC if it is recognised by an alternating machine using $f$ states,
where $f = o(\log(n))$ (respectively $f = o(n)$).

%% file: turing_machines.tex
A Turing machine is:
\begin{itemize}
	\item \textit{one-way} if its input head never moves to the left,
	\item \textit{realtime} if there exists a bound $K$ such that its input head can stay at the same place for at most $K$ steps.
\end{itemize}
It has been observed~\cite{Vollmer99} that one-way alternating Turing machines using at least logarithmic space
can simulate unrestricted alternating Turing machines.
This idea does not work for one-way realtime Turing machines.

\vskip1em
Given an alternating one-way realtime Turing machine, one can construct an alternating machine simulating it
as follows: a state of the machine is a configuration of the Turing machine, \textit{i.e.}
a tuple describing the control state, the content of the working tapes,
and the positions of the input and working heads.
This leads to the following well-known result.

\begin{theorem}
\label{thm:turing_machines}
An alternating one-way realtime Turing machine using space $f$ can be simulated by an alternating machine using $2^f$ states.
\end{theorem}

Recall that in the definition of OSC, we count the number of states, whereas when measuring space for Turing machines, 
we count how many bits are required to describe the configurations.
It is well known that these two quantities are exponentially related: $k$ bits allow to describe $2^k$ states,
and $n$ states require $\log(n)$ bits to be described.

It follows that all lower bounds for OSC imply lower bounds for one-way realtime Turing machines.


%% file: related_works.tex
The definition of online state complexity is due to Karp~\cite{Karp67},
and the first result proved in this paper is that non-regular languages have at least linear deterministic OSC.
Hartmanis and Shank considered the language of prime numbers written in binary, and showed in~\cite{HartmanisShank69} that 
it does not have subexponential deterministic OSC.
We pursue this question in this paper by consider the alternating OSC of the prime numbers.
Recently, we investigated the OSC of probabilistic automata;
we substantiated a claim by Rabin~\cite{Rabin63},
by exhibiting a probabilistic automaton which does not have subexponential deterministic OSC~\cite{Fijalkow16}.

\vskip1em
Three models of computations share some features with alternating OSC.
The first is boolean circuits; as explained in~\cite{Fijalkow16}, the resemblance is only superficial as circuits do not process the input in an online fashion.
For instance, one can observe that the language \text{Parity}, which is hard to compute with a circuit (not in $\text{AC}^0$ for instance), 
is actually a regular language, so trivial with respect to OSC.

The second model gives rise to the notion of automaticity; it has been introduced and studied by Shallit and Breitbart~\cite{SB96}.
The automaticity of a language $L$ is the function $\N \to \N$ which associates to $n$ the size of the smallest automaton
which agrees with $L$ on all words of length at most~$n$.
The essential difference is that automaticity is a non-uniform notion, as there is a different automaton for each $n$,
whereas OSC is uniform, as it considers one machine.
For this reason, the two measures behave completely differently. 
As an argument, consider a language $L$, and define its exponential padding: $\mathrm{Pad}(L) = \set{u \sharp^{2^{|u|}} \mid u \in L}$.
It is easy to see that for every language $L$, its exponential padding $\mathrm{Pad}(L)$ has exponential deterministic automaticity.
On the other hand, the OSC of $L$ and of $\mathrm{Pad}(L)$ are essentially the same.

\vskip1em
The third model is alternating communication complexity, developed by Babai, Frankl and Simon in 1986~\cite{BabaiFranklSimon86}.
In this setting, Alice has an input $x$ in $A$, Bob an input $y$ in $B$, and they want to determine $h(x,y)$ for a given boolean function 
$h : A \times B \to \set{0,1}$ known by all.
Alice and Bob are referees in a game involving two players, Prover and Verifier, who both know the two inputs.
Prover and Verifier exchange messages, whose conformity to the inputs is checked by Alice and Bob.
The cost of the protocol is the number of bits exchanged.

One can obtain strong lower bounds by a classical reduction from deterministic OSC to deterministic communication complexity;
it it thus tempting to extend this to the alternating setting.
However, as we argue in the following example, this gives very loose lower bounds, and our lower bound technique will be much stronger than this approach.

Consider the language $L$ studied in Subsection~\ref{subsec:lower_bound_example}, consisting of words of the form $u \sharp u_1 \sharp \cdots \sharp u_k$
such that $u$ is equal to $u_j$ for some $j$ in $\set{1,\ldots,k}$.
This induces an alternating communication complexity problem where Alice has a word $u$ of length $n$ as input, Bob has a word $v = u_1 \sharp \cdots \sharp u_k$,
and they want to determine whether $u \sharp v$ is in $L$.
One can show that if $L$ is in $\Alt{f}$, then the alternating communication complexity problem above can be solved by exchanging at most $f(n) \log(f(n))$ bits.
Here is a protocol: Prover first produces the number $j$ in $\set{1,\ldots,k}$, then Verifier enquires about a position $i$ in $\set{1,\ldots,n}$, 
and Prover replies by the letter in this position for both $u$ and $u_j$. 
At the end of this interaction, Alice and Bob each check that the declared letter is correct.
This protocol uses roughly $\log(k) + \log(n)$ bits, and since $k$ is at most $2^n$ this gives roughly $n + \log(n)$ bits.
This reasoning implies a lower bound on $f$, namely $f(n) \log(f(n)) \ge n + \log(n)$, from which we deduce that $f(n) \ge n$.
This is exponentially worse than the lower bound we will obtain using our technique, which states that $f(n) \ge 2^n$.

%% file: complexity_classes.tex

Denote $\Reg$ the class of regular languages, \textit{i.e.} those recognised by finite automata.
Then $\Det{1} = \ND{1} = \Alt{1} = \Reg$, \textit{i.e.} a language has constant OSC if, and only if, it is regular.
Indeed, a machine which uses a constant number of states is essentially a finite automaton,
and deterministic, non-deterministic and alternating finite automata are known to be equivalent.

\vskip1em
We remark that $\Det{|A|^n}$ is the class of all languages.
Indeed, consider a language $L$, we construct a deterministic machine recognising $L$ using exponentially many states.
Its set of states is $A^*$, the initial state is $\varepsilon$ and the transition function is defined by $\delta(w,a) = wa$.
The set of accepting states is simply $L$ itself.
The number of different states reachable by all words of length at most $n$ is the number of words of length at most $n$, 
\textit{i.e.} $\frac{|A|^{n+1} - 1}{|A| - 1}$.

It follows that the maximal OSC of a language is exponential,
and the online state complexity classes are relevant for functions smaller than exponential.
In particular, the class of languages of polynomial OSC is central.

\vskip1em
\noindent
We now give three examples.

\vskip1em
Denote $\Eq_3 = \set{w \in \set{a,b,c}^* \mid |w|_a = |w|_b = |w|_c}$.
The notation $|w|_a$ stands for the number of occurrences of the letter $a$ in $w$.
We construct a deterministic machine recognising this language using quadratically many states.
It has two counters that take integers values, initialised to $0$ each,
which maintain the value $(|w|_a - |w|_b, |w|_a - |w|_c)$.
To this end, the letter $a$ acts as $(+1,+1)$,
the letter $b$ as $(-1,0)$, the letter $c$ as $(0,-1)$.

Formally, the set of states is $\Z^2$ and the initial and only accepting state is $(0,0)$. 
The transitions are:
$$\begin{array}{l}
\delta((p,q),a) = (p+1,q+1) \\
\delta((p,q),b) = (p-1,q) \\
\delta((p,q),c) = (p,q-1)
\end{array}$$
After reading the word $w$, the machine is in the state $(|w|_a - |w|_b, |w|_a - |w|_c)$.
This means for a word of length at most $n$, there are at most $(2n+1)^2$ different states,
implying that $\Eq_3$ is in $\Det{n^2}$.
 
\vskip1em
Denote $\NotEq = \set{u \sharp v \mid u,v \in \set{0,1}^*, u \neq v}$.
Note that there are three ways to have $u \neq v$: either $v$ is longer than $u$,
or $v$ is shorter than $u$, or there exists a position for which they do not carry the same letter.
We construct a non-deterministic machine recognising this language using linearly many states.

We focus on the third possibility for the informal explanation.
The machine guesses a position in the first word, stores its position $p$ and its letter $a$,
and checks whether the corresponding position in the second word indeed carries a different letter.
To this end, after reading the letter $\sharp$, it decrements the position until reaching $1$, and checks
whether the letter is indeed different from the letter he carries in his state.

Formally, the set of states is $\N \times (A \cup \set{\bot}) \times \set{<,>} \cup \set{\top}$.
The first component carries a position, the second component a letter or $\bot$, meaning yet undeclared,
and the third component states whether the separator $\sharp$ has been read ($>$) or not ($<$).
The initial state is $(0,\bot,<)$.
The set of accepting states is $\set{\top} \cup \set{(p,\bot,>) \mid p \neq 0}$.
The transitions are:
$$\begin{array}{ll}
\delta((p,\bot,<),a) = (p+1,a,<) \vee (p+1,\bot,<) \\
\delta((p,a,<),b) = (p,a,<) \\
\delta((p,a,<),\sharp) = (p,a,>) \\
\delta((p,\bot,<),\sharp) = (p,\bot,>) \\
\delta((p,a,>),b) = (p-1,a,>) & \textrm{ if } p \neq 0 \\
\delta((1,a,>),b) = \top & \textrm{ if } a \neq b \\
\delta((p,\bot,>),a) = (p-1,\bot,>) \\
\delta((0,\bot,>),a) = \top \\
\delta(\top,a) = \top
\end{array}$$
If $v$ is longer than $u$, the state $\top$ will be reached using the second to last transition.
If $v$ is shorter than $u$, a state $(p,\bot,>)$ will be reached with $p \neq 0$.
If $u$ and $v$ have the same length and differ on some positions, the state $\top$ will be reached by guessing one such position.

After reading a word of length at most $n$, the machine can be in one of $(n+1) \cdot (|A| + 1) \cdot 2 + 1$ states,
thus $\NotEq$ is in $\ND{n}$.

\vskip1em
Denote $\Leq = \set{u \sharp v \mid u,v \in \set{0,1}^*, u \le_{\text{lex}} v}$.
We construct an alternating machine recognising this language using linearly many states by unravelling the inductive definition of the lexicographic order: 
$u \le_{\text{lex}} v$ if, and only if,
\[
(u_0 = 0 \wedge v_0 = 1) \vee (u_0 = v_0 \wedge u_{\mid \ge 1} \le_{\text{lex}} v_{\mid \ge 1}).
\]
The set of states is 
\[
(\text{lex} \times \N) \cup (\text{check-eq} \times \N \times \set{0,1} \times \set{<,>}) \cup (\text{check-sm} \times \N \times \set{<,>}) \cup \set{\top}.
\]
The initial state is $(\text{lex},0)$, the final state is $\top$.
The transitions are:
$$\begin{array}{lll}
\delta((\text{lex}),p),0) & = (\text{check-sm},p,<) \\
& \vee\ ((\text{check-eq},p,0,<) \wedge (\text{lex},p+1)) \\
\delta((\text{lex}),p),1) & = (\text{check-eq},p,1,<) \wedge (\text{lex},p+1) \\
\delta((\text{lex}),p),\sharp) & = \top \\

\delta((\text{check-eq}),p,a,<),b) & = (\text{check-eq},p,a,<) \\
\delta((\text{check-eq}),p,a,<),\sharp) & = (\text{check-eq},p,a,>) \\
\delta((\text{check-eq}),p,a,>),b) & = (\text{check-eq},p-1,a,>) & \text{for } p \neq 0\\
\delta((\text{check-eq}),0,a,>),b) & = \top & \text{if } a = b\\

\delta((\text{check-sm}),p,<),b) & = (\text{check-sm},p,<) \\
\delta((\text{check-sm}),p,<),\sharp) & = (\text{check-eq},p,>) \\
\delta((\text{check-sm}),p,>),b) & = (\text{check-sm},p-1,>) & \text{for } p \neq 0\\
\delta((\text{check-sm}),0,>),1) & = \top\\

\delta(\top,a) & = \top
\end{array}$$
After reading a word of length at most $n$, the machine can be in $8 n$ different states,
thus $\Leq$ is in $\Alt{n}$.

%% file: lower_bound.tex
Let $L$ be a language and $u$ a word. The left quotient of $L$ with respect to $u$ is
$$u^{-1} L = \set{v \mid uv \in L}.$$
If $u$ has length at most $n$, then we say that $u^{-1} L$ is a left quotient of $L$ of order $n$.


%

A lattice of languages is a set of languages closed under union and intersection.
Given a family of languages, the lattice it generates is the smallest lattice containing this family.

\begin{theorem}
\label{thm:lower_bound}
Let $L$ in $\Alt{f}$.
There exists a constant $C$ such that for all $n \in \N$,
there exists a family of $C \cdot f(n)$ languages whose generated lattice contains all the left quotients of $L$ of order $n$.
\end{theorem}

\begin{proof}
Let $\A$ be an alternating machine using recognising $L$ witnessing that $L$ is in $\Alt{f}$.

Fix $n$. Denote $Q_n$ the set of states reachable by some word of length at most $n$;
by assumption $|Q_n|$ is at most $C \cdot f(n)$ for some constant $C$.
For $q$ in $Q_n$, denote $L(q)$ the language recognised by $\A$ taking $q$ as initial state,
and $\L_n$ the family of these languages.

We prove by induction over $n$ that all left quotients of $L$ or order $n$ can be obtained as boolean combinations of languages in $\L_n$.

The case $n = 0$ is clear, as $\varepsilon^{-1} L = L = L(q_0)$.

Consider a word $w$ of length $n+1$, denote $w = ua$.
We are interested in $w^{-1} L = a^{-1} (u^{-1} L)$, so let us start by considering $u^{-1} L$.
By induction hypothesis, $u^{-1} L$ can be obtained as a boolean combination of languages in $\L_n$:
denote $u^{-1} L = \phi(\L_n)$, meaning that $\phi$ is a boolean formula whose atoms are languages in $\L_n$.

Now consider $a^{-1} \phi(\L_n)$.
Observe that the left quotient operation respects both unions and intersections,
\textit{i.e.} $a^{-1}(L_1 \cup L_2) = a^{-1} L_1 \cup a^{-1} L_2$ and $a^{-1}(L_1 \cap L_2) = a^{-1} L_1 \cap a^{-1} L_2$.
It follows that $w^{-1} L = a^{-1} (\phi(\L_n)) = \phi(a^{-1} \L_n)$; 
this notation means that the atoms are languages of the form $a^{-1} M$ for $M$ in $\L_n$,
\textit{i.e.} $a^{-1} L(q)$ for $q$ in $S_n$.

To conclude, we remark that $a^{-1} L(q)$ can be obtained as a boolean combination of the languages $L(p)$,
where $p$ are the states that appear in $\delta(q,a)$.
To be more precise, we introduce the notation $\psi(L(\cdot))$, on an example: if $\psi = p \wedge (r \vee s)$,
then $\psi(L(\cdot)) = L(p) \wedge (L(r) \vee L(s))$.
With this notation, $a^{-1} L(q) = \delta(a,q)(L(\cdot))$.
Thus, for $q$ in $Q_n$, we have that $a^{-1} L(q)$ can be obtained as a boolean combination of languages in $\L_{n+1}$.

Putting everything together, it implies that $w^{-1} L$ can be obtained as a boolean combination of languages in $\L_{n+1}$,
finishing the inductive proof.
\hfill\qed
\end{proof}

%% file: query_tables.tex
\begin{definition}[Query Table]
Consider a family of languages $\L$.
Given a word $w$, its profile with respect to $\L$, or $\L$-profile,
is the boolean vector stating whether $w$ belongs to $L$, for each $L$ in $\L$.
The size of the query table of $\L$ is the number of different $\L$-profiles, when considering all words.

For a language $L$, its query table of order $n$ is the query table of the left quotients of $L$ of order $n$.
\end{definition}

The name query table comes from the following image: 
the query table of $\L$ is the infinite table whose columns are indexed by languages in $\L$ and rows by words (so, there are infinitely many rows).
The cell corresponding to a word $w$ and a language $L$ in $\L$ is the boolean indicating whether $w$ is in $L$.
Thus the $\L$-profile of $w$ is the row corresponding to $w$ in the query table of $\L$.

\begin{lemma}
\label{lem:upper_bound_query_table}
Consider a lattice of languages $\L$ generated by $k$ languages.
The query table of $\L$ has size at most $2^k$.
\end{lemma}

Indeed, there are at most $2^k$ different profiles with respect to $\L$.

\begin{theorem}
\label{thm:query_table}
Let $L$ in $\Alt{f}$.
There exists a constant $C$ such that for all $n \in \N$,
the query table of $L$ of order $n$ has size at most $2^{C \cdot f(n)}$.
\end{theorem}

Thanks to Theorem~\ref{thm:query_table}, to prove that $L$ does not have sublogarithmic (respectively sublinear) alternating OSC, 
it is enough to exhibit a constant $C > 0$ 
such that for infinitely many $n$, the query table of $L$ of order $n$ has size at least $C \cdot n$
(respectively at least $2^{C \cdot n}$).

The proof of Theorem~\ref{thm:query_table} relies on the following lemma.

\begin{lemma}
\label{lem:query_table}
Consider two lattices of languages $\L$ and $\M$.
If $\M \subseteq \L$, then the size of the query table of $\M$ is smaller than or equal to the size of the query table of $\L$.
\end{lemma}

\begin{proof}
It suffices to observe that the query table of $\M$ is ``included'' in the query table of $\L$.
More formally, consider in the query table of $\L$ the sub-table which consists of rows corresponding to languages in $\M$:
this is the query table of $\M$.
This implies the claim.
\hfill\qed
\end{proof}

We now prove Theorem~\ref{thm:query_table}.
Thanks to Theorem~\ref{thm:lower_bound}, the family of left quotients of $L$ of order $n$ is contained in 
a lattice generated by a family of size at most $C \cdot f(n)$.
It follows from Lemma~\ref{lem:query_table} that the size of the query table of $L$ of order $n$ 
is smaller than or equal to the size of the query table of a lattice generated by at most $C \cdot f(n)$ languages,
which by Lemma~\ref{lem:upper_bound_query_table} is at most $2^{C \cdot f(n)}$.

%% file: non-det-cfl.tex
As a first application of our technique, we exhibit a language which has maximal (\textit{i.e.} exponential) alternating OSC.
Surprisingly, this language is simple in the sense that it is context-free and definable in Presburger arithmetic.

We say that $L$ has subexponential alternating OSC if $L \in \Alt{f}$ for some $f$ such that $f = o(C^n)$ for all $C > 1$.
Thanks to Theorem~\ref{thm:query_table}, to prove that $L$ does not have subexponential alternating OSC, it is enough to exhibit a constant $C > 1$ 
such that for infinitely many $n$, the query table of the left quotients of $L$ of order $n$ has size at least $2^{C^n}$.

\begin{theorem}\label{thm:exp_alt}
There exists a language which does not have subexponential alternating OSC,
yet is both context-free and definable in Presburger arithmetic.
\end{theorem}

\begin{proof}
Denote 
$$L = \set{u \sharp u_1 \sharp u_2 \sharp \cdots \sharp u_k 
\left| 
\begin{array}{c}
u,u_1,\ldots,u_k \in \set{0,1}^*, \\
\exists j \in \set{1,\ldots,k}, u = \overline{u_j}
\end{array}
\right.}.$$
The notation $\overline{u}$ stands for the reverse of $u$: formally, $\overline{u} = u(n-1) \cdots u(0)$.

It is easy to see that $L$ is both context-free and definable in Presburger arithmetic 
(the use of reversed words in the definition of $L$ is only there to make $L$ context-free).

We show that $L$ does not have subexponential alternating OSC.
We prove that for all $n$, the query table of the left quotients of $L$ of order $n$ has size at least $2^{2^n}$.
Thanks to Theorem~\ref{thm:query_table}, this implies the result.

Fix $n$. Denote by $U$ the set of all words $u$ in $\set{0,1}^n$, it has cardinal $2^n$.
Consider any subset $S$ of $U$, we argue that there exists a word $w$ which
satisfies that if $u$ in $U$, then the following equivalence holds:
$$w \in u^{-1} L \Longleftrightarrow u \in S.$$
This shows the existence of $2^{2^n}$ different profiles with respect to the left quotients of order $n$, as claimed.

Denote $u_1,\ldots,u_{|S|}$ the words in $S$.
Consider 
$$w = \sharp \overline{u_1} \sharp \overline{u_2} \sharp \cdots \sharp \overline{u_{|S|}}.$$
The word $w$ clearly satisfies the claim above.
\hfill\qed
\end{proof}

%% file: alt.tex
\begin{theorem}
For each $\ell \ge 2$, there exists a language $L_\ell$ such that:
\begin{itemize}
	\item $L_\ell$ is in $\Alt{n^\ell}$,
	\item $L_\ell$ is not in $\Alt{n^{\ell - \varepsilon}}$ for any $\varepsilon > 0$.
\end{itemize}
\end{theorem}

Consider the alphabet $\set{0,1} \cup \set{\lozenge,\sharp}$.

Let $\ell \ge 2$. Denote 
$$L_\ell = \set{\lozenge^p u \sharp u_1 \sharp u_2 \sharp \cdots \sharp u_k \left|
\begin{array}{c}
u,u_1,\ldots,u_k \in \set{0,1}^*,\\
j \le p^\ell \textrm{ and } u = u_j
\end{array}\right.}.$$

\begin{proof}\hfill
\begin{itemize}
	\item The machine has three consecutive phases:
\begin{enumerate}
	\item First, a non-deterministic guessing phase while reading $\lozenge^p$, which passes onto the second phase a number $j$ in $\set{1,\ldots,p^\ell}$.

	\vskip1em
	Formally, the set of states for this phase is $\N$, the initial state is $0$ and the transitions are:
	$$\begin{array}{l}
	\delta(0,\lozenge) = 1 \\
	\delta(k^\ell,\lozenge) = \bigvee_{j \in \set{1,\ldots,(k+1)^\ell}} j \\
	\delta(p,\lozenge) = p
	\end{array}$$
	
	\item Second, a universal phase while reading $u$.
	For each $i$ in $\set{1,\ldots,|u|}$, the machine launches one copy storing the position $i$, the letter $u(i)$ and the number $j$ guessed in the first phase.

	\vskip1em
	Formally, the set of states for this phase is $\N \times (\set{0,1} \cup \set{\bot}) \times \N$.
	The first component is the length of the word read so far (in this phase), the second component stores the letter stored, 
	where the letter $\bot$ stands for undeclared, and the last component is the number $j$.

	The initial state is $(0,\bot,j)$.
	The transitions are:
	$$\begin{array}{l}
	\delta((q,\bot,j),a) = (q+1,\bot,j) \wedge (q,a,j) \\
	\delta((q,a,j),b) = (q,a,j)	
	\end{array}$$

	This requires quadratically many states.
	
	\item Third, a deterministic phase while reading $\sharp u_1 \sharp u_2 \sharp \cdots \sharp u_k$.
	It starts from a state of the form $(q,a,j)$.
	It checks whether $u_j(q) = a$.
	The localisation of the $u_j$ is achieved by decrementing the number $j$ by one each time a letter $\sharp$ is read.
	While in the corresponding $u_j$, the localisation of the position $q$ in $u_j$ as achieved by decrementing one position at a time.

	This requires quadratically many states.
\end{enumerate}

	\item We now prove the lower bound.
	
We prove that for all $n$, the size of the query table of $L_\ell$ of order $n + 2^{\frac{n}{\ell}}$ is at least $2^{2^n}$.
Thanks to Theorem~\ref{thm:query_table}, this implies that $L_\ell$ is not in $\Alt{n^{\ell - \varepsilon}}$ for any $\varepsilon > 0$.

Fix $n$. 
Denote by $U$ the set of all words $u$ in $\set{0,1}^n$, it has cardinal $2^n$.

Observe that $\lozenge^{2^{\frac{n}{\ell}}} u \sharp u_1 \sharp u_2 \sharp \cdots \sharp u_{2^n}$ belongs to $L_\ell$
if, and only if, there exists $j$ in $\set{1,\ldots,2^n}$ such that $u = u_j$.

Consider any subset $S$ of $U$, we argue that there exists a word $w$ which
satisfies that if $u$ in $U$, then the following equivalence holds:
$$w \in \left( \lozenge^{2^{\frac{n}{\ell}}}u \right)^{-1} L \Longleftrightarrow u \in S.$$
This shows the existence of $2^{2^n}$ different profiles with respect to the left quotients of order $n + 2^{\frac{n}{\ell}}$, as claimed.

Denote $u_1,\ldots,u_{|S|}$ the words in $S$.
Consider 
$$w = \sharp \overline{u_1} \sharp \overline{u_2} \sharp \cdots \sharp \overline{u_{|S|}}.$$
The word $w$ clearly satisfies the claim above.
\end{itemize}
\end{proof}

%% file: prime.tex
In this section, we give lower bounds on the alternating online state complexity of the language of prime numbers written in binary:
$$\Prime = \set{u \in \set{0,1}^* \mid \bin(u) \textrm{ is prime}}.$$
By definition $\bin(w) = \sum_{i \in \set{0,\ldots,n-1}} w(i) 2^i$;
note that the least significant digit is on the left.

The complexity of this language has long been investigated; many efforts have been put in finding upper and lower bounds.
In 1976, Miller gave a first conditional polynomial time algorithm, assuming the generalised Riemann hypothesis~\cite{Miller76}.
In 2002, Agrawal, Kayal and Saxena obtained the same results, but non-conditional, 
\textit{i.e.} not predicated on unproven number-theoretic conjectures~\cite{AKS02}.

\vskip1em
The first lower bounds were obtained by Hartmanis and Shank in 1968,
who proved that checking primality requires at least logarithmic deterministic space~\cite{HartmanisShank68},
conditional on number-theoretic assumptions.
It was shown by Hartmanis and Berman in 1976 that if the number is presented in unary, 
then logarithmic deterministic space is necessary and sufficient~\cite{HartmanisBerman76}.

The best lower bound we know from circuit complexity is due to Allender, Saks and Shparlinski: they proved unconditionally in 2001 that $\Prime$ 
is not in $\mathrm{AC}^0[p]$ for any prime $p$~\cite{ASS01}.

\vskip1em
The results above are incomparable to our setting, as we are here interested in online computation.
The first and only result to date about the OSC of $\Prime$ is due to Hartmanis and Shank in 1969:

\begin{theorem}[\cite{HartmanisShank69}]
\label{thm:hs}
The set of prime numbers written in binary does not have subexponential deterministic online state complexity.
\end{theorem}

Their result is unconditional, and makes use of Dirichlet's theorem on arithmetic progressions of prime numbers.
More precisely, they prove the following result.

\begin{proposition}[\cite{HartmanisShank69}]
\label{prop:hs}
Fix $n > 1$, and consider $u$ and $v$ two differents words of length $n$ starting with a $1$.
Then the left quotients $u^{-1} \Prime$ and $v^{-1} \Prime$ are different.
\end{proposition}

Proposition~\ref{prop:hs} directly implies Theorem~\ref{thm:hs}~\cite{HartmanisShank69}.
It also yields a lower bound of $n-1$ on the size of the query table of $\Prime$ of order $n$.
Thus, together with Theorem~\ref{thm:query_table}, this proves that $\Prime$ does not have sublogarithmic alternating OSC.

\begin{corollary}
The set of prime numbers written in binary does not have sublogarithmic alternating online state complexity.
\end{corollary}

Our contribution in this section is to extend this result by showing that $\Prime$ does not have sublinear alternating OSC,
which is an exponential improvement.

\begin{theorem}
\label{thm:prime}
The set of prime numbers written in binary does not have sublinear alternating online state complexity.
\end{theorem}

We state the following immediate corollary for Turing machines, relying on Theorem~\ref{thm:turing_machines}.

\begin{corollary}
The set of prime numbers written in binary cannot be recognised by a one-way realtime alternating Turing machine using sublogarithmic space.
\end{corollary}

Our result is unconditional, but it relies on the following advanced theorem from number theory, 
which can be derived from the results obtained by Maier and Pomerance~\cite{MP90}.
Note that their results are more general; we simplified the statement to make it both simple and closer to what we actually use.

Simply put, this result says that in any (reasonable) arithmetic progression and for any degree of isolation, 
there exists a prime number in this progression, isolated with respect to all prime numbers.

\begin{theorem}[\cite{MP90}]
\label{thm:maier_pomerance}
For every arithmetic progression $a + b \N$ such that $a$ and $b$ are coprime, for every $N$,
there exists a number $k$ such that $p = a + b \cdot k$ is the only prime number in $[p-N,p+N]$.
\end{theorem}

We proceed to the proof of Theorem~\ref{thm:prime}.

\begin{proof}
We show that for all $n > 1$, the query table of $\Prime$ of order $n$ has size at least $2^{n-1}$.
Thanks to Theorem~\ref{thm:query_table}, this implies the result.

Fix $n > 1$. Denote by $U$ the set of all words $u$ of length $n$ starting with a $1$. 
Equivalently, we see $U$ as a set of numbers; it contains all the odd numbers smaller than $2^n$.
It has cardinal $2^{n-1}$.

We argue that for all $u$ in $U$, there exists a word $w$ which satisfies that for all $v$ in $U$, 
$w$ is in $v^{-1} \Prime$ if, and only if, $u = v$.
In other words the profile of $w$ is $0$ everywhere but on the column $u^{-1} \Prime$.
Let $u$ in $U$; denote $a = \bin(u)$.
Consider the arithmetic progression $a + 2^n \N$; note that $a$ and $2^n$ are coprime.
Thanks to Theorem~\ref{thm:maier_pomerance}, for $N = 2^n$, there exists 
a number $k$ such that $p = a + 2^n \cdot k$ is the only prime number in $[p-N,p+N]$.
Denote $w$ a word such that $\bin(w) = k$.
We show that for all $v$ in $U$, we have the following equivalence: $w$ is in $v^{-1} \Prime$ if, and only if, $u = v$.

Indeed, $\bin(vw) = \bin(v) + 2^n \cdot \bin(w)$. 
Observe that 
$$|\bin(vw) - \bin(uw)| = |\bin(v) - \bin(u)| < 2^n.$$
Since $p$ is the only prime number in $[p-2^n,p+2^n]$, the equivalence follows.

We constructed $2^{n-1}$ words each having a different profile, implying the claimed lower bound.
\hfill\qed
\end{proof}

Theorem~\ref{thm:prime} proves a linear lower bound on the alternating OSC of $\Prime$.
We do not know of any non-trivial upper bound, and believe that there are none, meaning that $\Prime$ does not have subexponential alternating OSC.

An evidence for this is the following probabilistic argument. Consider the distribution of languages over $\set{0,1}^*$
such that a word $u$ in thrown into the language with probability $\frac{1}{|u|}$.
It is a common (yet flawed) assumption that the prime numbers satisfy this distribution, as witnessed for instance by the Prime Number theorem.
One can show that with high probability such a language does not have subexponential alternating OSC,
the reason being that two different words are very likely to induce different profiles in the query table.
Thus it is reasonable to expect that $\Prime$ does not have subexponential alternating OSC.

\vskip1em
We dwell on the possibility of proving stronger lower bounds for the alternating OSC of $\Prime$.
Theorem~\ref{thm:maier_pomerance} fleshes out the \textit{sparsity} of prime numbers:
it constructs isolated prime numbers in any arithmetic progression,
and allows us to show that the query table of $\Prime$ contains all profiles with all but one boolean value set to false.

To populate the query table of $\Prime$ further, one needs results witnessing the \textit{density} of prime numbers,
\textit{i.e.} to prove the existence of clusters of prime numbers.
This is in essence the contents of the Twin Prime conjecture, or more generally of Dickson's conjecture,
which are both long-standing open problems in number theory,
suggesting that proving better lower bounds is a very challenging objective.
The Dickson's conjecture reads:

\begin{conjecture}[Dickson's Conjecture]
Fix $b$ and a subset $S \subseteq \set{1,\ldots,b-1}$ such that 
there exists no prime number $p$ which divides $\prod_{a \in S} (b \cdot k + a)$ for every $k$ in $\N$.
Then there exists a number $k$ such that $b \cdot k + a$ is prime for each $a$ in $S$.
\end{conjecture}

Unfortunately, the Dickson's conjecture is not enough; indeed, it constructs words whose profiles contain at least a certain number
of booleans set to true (the ones from the subset $S$).
To obtain different profiles, one needs to ensure that these profiles contain exactly these booleans set to true.
Hence the following slight generalisation of Dickson's conjecture.

\begin{conjecture}
\label{conjecture:gen}
Fix $b$ and a subset $S \subseteq \set{1,\ldots,b-1}$ such that 
there exists no prime number $p$ which divides $\prod_{a \in S} (b \cdot k + a)$ for every $k$ in $\N$.
Then there exists a number $k$ such that $b \cdot k + a$ is prime if, and only if, $a$ is in $S$.
\end{conjecture}

\begin{theorem}
Assuming the Conjecture~\ref{conjecture:gen} holds true, 
the set of prime numbers written in binary does not have subexponential alternating online state complexity.
\end{theorem}

\begin{proof}
We show that for all $n > 1$, the query table of $\Prime$ of order $n$ has size doubly-exponential in $n$.
Thanks to Theorem~\ref{thm:query_table}, this implies the result.

Fix $n > 1$. As above, denote by $U$ the set of all words $u$ of length $n$ starting with a $1$, \textit{i.e.} odd numbers. 

For a subset $S$ of $U$, denote $(\lozenge)$ the property that 
there exists no prime number $p$ which divides $\prod_{a \in S} (b \cdot k + a)$ for every $k$ in $\N$.

Let $S$ be a subset of $U$ satisfying $(\lozenge)$.
We argue that there exists a word $w$ which satisfies that for all $v$ in $U$, 
$w$ is in $v^{-1} \Prime$ if, and only if, $v$ is in $S$.
In other words the profile of $w$ is $1$ on the columns corresponding to $S$, and $0$ everywhere else.

Thanks to Conjecture~\ref{conjecture:gen}, there exists a number $k$ such that 
$2^n \cdot k + a$ is prime if, and only if $a$ is in $S$.
Denote $w$ a word such that $\bin(w) = k$, it clearly satisfies the condition above.

For each subset $S$ satisfying $(\lozenge)$ we constructed a word such that these words have pairwise different profiles.
To conclude, we need to explain why there are doubly-exponentially many subsets $S$ satisfying $(\lozenge)$.
Advanced sieve techniques from number theory can be used to precisely estimate this;
we will rely on a simpler argument here to obtain a doubly-exponential lower bound.
Indeed, observe that if there exists one subset $S$ satisfying $(\lozenge)$ of exponential size (\textit{i.e.} $C^n$ for some $C > 1$), 
then all subsets of $S$ also satisfy $(\lozenge)$, which yields doubly-exponentially many of them.
We claim that $S$ defined by $\set{a \in U \mid 2^n + a \text{ is a prime number}}$ satisfies $(\lozenge)$.
This follows from the remark that no prime number can divide both $\prod_{a \in S} a$ and $\prod_{a \in S} (2^n + a)$.
\hfill\qed
\end{proof}